\documentclass[a4paper]{article}

\usepackage{amsmath,psfrag,stmaryrd}
\usepackage[mathscr]{eucal}
\usepackage{amssymb,color}
\usepackage{graphicx}
\usepackage{latexsym,nicefrac}
\usepackage{tikz}
\usepackage{amsthm}

\usepackage{authblk}
\usetikzlibrary{shapes}
\usepackage{fancyhdr}

\newcommand{\odd}[1]{\textup{\textsf{Odd}}({#1})}

\newcommand{\codd}[1]{\textup{\textsf{Odd}}[{#1}]}

\newcommand{\ket}[1]{\left| #1 \right\rangle}
\newcommand{\bra}[1]{\left\langle #1 \right|}

\newcommand\restrict[1]{\raisebox{-.5ex}{$|$}_{#1}}

\newcommand\minus{%
  \setbox0=\hbox{-}%
  \vcenter{%
    \hrule width\wd0 height \the\fontdimen8\textfont3%
  }%
}

\renewcommand{\phi}{\varphi}
\newcommand{\comp}[1]{{#1}^c}

\newcommand{\href}[1]{$#1$}
\newcommand{\url}[1]{$#1$}
\newcommand{\path}[1]{$#1$}

\newtheorem{theorem}{Theorem}
\newtheorem{property}{Property}
\newtheorem{lemma}[theorem]{Lemma}

\theoremstyle{definition}
\newtheorem{definition}[theorem]{Definition}

\theoremstyle{remark}
\newtheorem*{remark}{Remark}

\begin{document}

\title{Determinism and Computational Power of Real Measurement-based Quantum Computation}

\author[1]{Simon Perdrix \thanks{\texttt{simon.perdrix@loria.fr}}}
\author[2]{Luc Sanselme \thanks{\texttt{luc.sanselme@loria.fr}}}
\affil[1]{\small CNRS, LORIA, Universit\'e de Lorraine, Inria Project Team Carte, , France}
\affil[2]{\small LORIA, CNRS, Universit\'e de Lorraine, Inria Project Team Caramba, Lyc\'ee Poincar\'e, France}
%
%

\pagestyle{fancy}
\lhead{Determinism and Computational Power of Real MBQC}
\rhead{S. Perdrix \& L. Sanselme}
\date{}

\maketitle

\begin{abstract}

Measurement-based quantum computing (MBQC) is a universal model for quantum computation.
The combinatorial characterisation of determinism in this model, powered by measurements, and hence, fundamentally probabilistic, is the cornerstone of most of the breakthrough results in this field. 
The most general known sufficient condition for a deterministic MBQC to be driven is that the underlying graph of the computation has a particular kind of flow called Pauli flow. The necessity of the Pauli flow was an open question. We show that Pauli flow is not necessary, providing several counter examples. We prove however that Pauli flow is necessary for determinism in the \emph{real} MBQC model, an interesting and useful fragment of MBQC. 

We explore the consequences of this result for real MBQC and its applications. Real MBQC and more generally real quantum computing is known to be universal for quantum computing. Real MBQC has been used for interactive proofs by McKague. 
The two-prover case corresponds to real-MBQC on bipartite graphs. While (complex) MBQC on bipartite  graphs are universal, the universality of real  MBQC on bipartite graphs was an open question. We show that real bipartite MBQC  is not universal proving that all measurements of real bipartite MBQC can be parallelised leading to constant depth computations. As a consequence, McKague's techniques cannot lead to two-prover interactive proofs.

\end{abstract}

\section{Introduction}

Measurement-based quantum computing \cite{RB01,RBB03} (MBQC for short) is a universal model for quantum computation. This model is not only very promising  in terms of the physical realisations of the quantum computer \cite{Petal,Wetal}, MBQC has also several theoretical advantages, e.g. parallelisation of quantum operations \cite{BKP10,BK07} (logarithmic separation with the traditional model of quantum circuits), blind quantum computing \cite{BFK08} (a protocol for  delegated quantum computing), fault tolerant quantum computing \cite{raussendorf2006fault}, simulation \cite{delfosse2015wigner},  contextuality \cite{raussendorf2013contextuality}, interactive proofs \cite{McKague,BFK08}. 

In MBQC, a computation consists of performing local quantum measurements over a large entangled resource state. The resource state is described by a graph -- using the so-called graph state formalism \cite{HEB04}. The \emph{tour de force} of this model is to tame the fundamental non-determinism of the quantum measurements: the number of possible outputs of a measurement-based computation on a given input is exponential in the number of measurements, and each of these {branches} of the computation is produced with an exponentially small probability. The only known technique to make such a fundamentally probabilistic computation exploitable is to implement a correction strategy which makes the overall computation deterministic: it does not affect the probability for each branch of the computation to occur, but it guarantees that all the branches produce the same output. 

The existence of a {correction} strategy relies on the structures of the entanglement of the quantum state on which the measurements are performed. 
Deciding whether a given resource state allows determinism is a central question in MBQC. Several sufficient conditions for determinism have been introduced. First in \cite{DK06} the notion of \emph{causal flow} has been introduced: if the graph describing the entangled resource state has a causal flow then a deterministic MBQC can be driven on this resource. Causal flow has been generalized to a weaker condition called Generalized flow (Gflow) which is also sufficient for determinism. Gflow has been proved to be necessary for a robust variant of determinism and when roughly speaking there is no  Pauli measurement, a special class of quantum measurements (see section \ref{sec:MBQC} for details) \cite{BKMP07}.  In the same paper, the authors have introduced a  weaker notion of flow called Pauli Flow, allowing some measurements to be Pauli measurements. 
Pauli flow is the weakest known sufficient condition for determinism and its necessity was a crucial open question as the characterisation of determinism in MBQC is the cornerstone of most of the applications of MBQC.

In section \ref{sec:MBQC}, we present the MBQC model, and the tools that come with it.
Our first contribution is to provide a simpler characterisation of the Pauli flow (Proposition \ref{prop:sim}), with three instead of nine conditions to satisfy for the existence of a Pauli flow. Our main contribution is to prove in section \ref{sec:robdet}  that the Pauli flow is not necessary in general -- by pointing out several counter examples --  but is actually necessary for \emph{real} MBQC (Theorem \ref{thm:paulinec}). Real MBQC is a restriction of MBQC where only real observables are used, i.e.~observables which eigenstates are quantum states that can be described using real numbers. Quantum mechanics, and hence models of quantum computation, are traditionally based on complex numbers. Real quantum computing is universal for quantum computation \cite{BV97a} and has been crucially used recently in the study of contextuality and simulation by means of quantum computing by state injection \cite{delfosse2015wigner}. Real MBQC \cite{MP13} may lead to several other applications. One of them is an interactive proof protocol built by McKague \cite{McKague}. McKague introduced a protocol where a verifier using a polynomial number of quantum provers can perform a computation, with the guaranty that, if a prover has cheated, it will be able to detect it. An open question left in \cite{McKague} by McKague is to know whether this model can bring to an interactive proof protocol with only two quantum provers. We answer negatively to this question in section \ref{subs:intpro}. Our third contribution is to point out the existence of a 
kind of supernormal-form 
for Pauli flow in real MBQC on bipartite graphs (Lemma \ref{lem:8}). This result enables us to prove in Theorem \ref{thm:constdepth} that real MBQC on bipartite graphs is not very powerful: all measurements of a real bipartite MBQC can be parallelised. As a consequence, only  problems that can be solved in constant depth can be solved using real bipartite MBQC.

\section{Measurement-based quantum computation, Generalized Flow and Pauli Flow}
\label{sec:MBQC}
\noindent{\bf Notations.} 
We assume the reader familiar with quantum computing notations, otherwise one can refer to Appendix \ref{QC} or to \cite{NC00}. We will use the following set/graph  notations: First of all, the \emph{symmetric difference} of two sets $A$ and $B$ will be denoted $A \Delta B := (A\cup B)\setminus (A\cap B)$. 
We will use intensively the \emph{open} and \emph{closed neighbourhood}. Given a simple undirected graph $G=(V,E)$, for any $u\in V$, $N(u) := \{v \in V~|~ (u,v)\in E\}$ is the (open) neighbourhood of $u$, and $N[u]:=N(u)\cup \{u\}$ is the closed neighbourhood of $u$. For any subset $A$ of $V$, $\odd A:= \Delta_{v\in A} N(v)$ (resp. $\codd A:= \Delta_{v\in A} N[u]$) is the odd  (resp. odd closed) neighbourhood of $A$. 
		Also, we will use the notion of \emph{extensive maps}. A map $ f:A\to 2^B$, with $A\subseteq B$ is extensive if the transitive closure of $\{(u,v): v\in  f(u)\}$ is a strict partial order. We say that $ f$ is extensive with respect to a strict partial order $\prec$ if $(v\in  f(u)\Rightarrow u\prec v)$.

\subsection{MBQC, concretely, abstractly}

In this section, a brief description of the measurement-based quantum computation is given, a more detailed introduction can be found in \cite{DKP07,DKPP09}. 
Starting from a low-level description of measurement-based quantum computation using the so-called patterns of the Measurement-Calculus -- an assembly language composed of 4 kinds of commands: creation of ancillary qubits, entangling operation, measurement and correction -- we end up with a graph theoretical description of the computation and in particular of the underlying entangled resource of the computation. 

\subsection{Measurement-Calculus patterns: an assembly language}

An assembly language for MBQC is the Measurement-Calculus \cite{DKP07,DKPP09}:  a pattern is a sequence of commands, each command is either:
\\-- $N_u$:  initialisation of a fresh qubit $u$ in the state $\ket +=\frac{\ket 0+\ket 1}{\sqrt 2}$; 
\\-- $E_{u,v}$ entangling two qubits $u$ and $v$ by applying Control-Z operation $\Lambda Z:\ket{x,y}\mapsto (-1)^{xy}\ket {x,y}$ to the qubits $u$ and $v$; 
\\-- $M_u^{\lambda_u,\alpha_u}$ measurement of qubit $u$ according to the observable $\mathcal O_{\lambda_u, \alpha_u}$ described below;
\\-- $X_u^{s_v}$ (resp. $Z_u^{s_v}$), a correction which consists of applying Pauli $X:\ket x\mapsto \ket {1-x}$ (resp. $Z:\ket x\mapsto (-1)^x\ket x$) to qubit $u$ iff $s_v$ (the classical outcome of the measurement of qubit $v$) is $1$. 

A pattern is subject to some basic well-formedness conditions like: no operation can be applied on a qubit $u$ after $u$ being measured; a correction cannot depend on a signal $s_u$ if qubit $u$ is not yet measured. 

The qubits which are not initialised using the $N$ command are the input qubits, and those which are not measured are the output qubits.  The measurement of a qubit $u$ is characterized by $\lambda_u\subset\{X,Y,Z\}$ a subset of one or two Pauli operators, and an angle 
$\alpha_u\in [0,2\pi)$: \\
-- when $\lambda_u=\{M\}$ is a singleton, $u$ is measured according to $\mathcal O_{\lambda_u, \alpha_u}:=M$ if $\alpha_u=0$ or $\mathcal O_{\lambda_u, \alpha_u}:=-M$ if $\alpha_u=\pi$. \\-- when $|\lambda_u|=2$, $u$ is measured in the $\lambda_u$-plane of the Bloch sphere with an angle $\alpha_u$, i.e.~according to the observable:
 $$\mathcal O_{\lambda_u, \alpha_u}:=\begin{cases}
 \cos(\alpha_u)X_u+\sin(\alpha_u)Y_u& \text{if $\lambda_u = \{X,Y\}$}\\
 \cos(\alpha_u)Y_u+\sin(\alpha_u)Z_u&\text{if $\lambda_u = \{Y,Z\}$}\\
\cos(\alpha_u)Z_u+\sin(\alpha_u)X_u& \text{if $\lambda_u = \{Z,X\}$}
 \end{cases}
 $$
 Measurement of qubit $u$ produces a classical outcome $(-1)^{s_u}$ where $s_u\in \{0,1\}$ is called \emph{signal}, or simply \emph{classical outcome} with a slight abuse of notation.

\subsection{A graph-based representation}

In the Measurement-Calculus, the patterns are equipped with an equational theory which captures some basic invariant properties, e.g.~two operations acting on distinct qubits commute, or $E_{u,v}$ is equivalent to $E_{v,u}$. It is easy to show using the equations of the Measurement-Calculus that  any 
pattern can be transformed into an equivalent pattern of the form:

$$ \left(\prod^\prec_{u\in \comp O} Z^{s_u}_{\mathtt z(u)}X^{s_u}_{\mathtt x(u)}M_u^{\lambda_u,\alpha_u}\right)\left(
 \prod_{(u,v)\in G}E_{u,v}\right)\left( \prod_{u\in \comp I}N_u\right)$$
where $G=(V,E)$ is a simple undirected graph, $I,O\subseteq V$ are respectively the input and output qubits, and $\mathtt x, \mathtt z: \comp{O} \to 2^{V}$ are two extensive maps, i.e.~the relation $\prec$ defined as the transitive closure of $\{(u,v): v\in \mathtt x(u) \cup \mathtt z(u)\}$  
is a strict partial order. Notice that $O^c:= V\setminus O$ and $X^{s_u}_{\mathtt x(u)} := \prod_{v\in\mathtt x(u)} X_v^{s_u}$. Moreover the product $\prod_{(u,v)\in G}$ means that the indices are the  edges of the $G$, in particular each edge is taken once. 
 
 The septuple $(G,I,O,\lambda, \alpha, \mathtt x, \mathtt z)$ is a graph-based representation which captures entirely the semantics of the corresponding pattern. We simply call an MBQC such a septuple.

 \subsection{Semantics and Determinism}
 
 An MBQC $(G,I,O,\lambda, \alpha, \mathtt x, \mathtt z)$ has a fundamentally probabilistic evolution with potentially $2^{|\comp O|}$ possible branches as the computation consists of $|\comp O|$ measurements. For any $s\in \{0,1\}^{|\comp O|}$, let $A_s :\mathbb C^{\{0,1\}^{I}}\to \mathbb C^{\{0,1\}^O}$ be  $$A_s(\ket \phi)= \left(\prod^\prec_{u\in \comp O} Z^{s_u}_{\mathtt z(u)}X^{s_u}_{\mathtt x(u)}\bra {\phi^{\lambda_u, \alpha_u}_{s_u}}_u\right)\left(
 \!\prod_{(u,v)\in G}\!\!\!\!\Lambda Z_{u,v}\right)\left(\ket \phi\otimes \frac{\sum_{x\in \{0,1\}^{I^c}} \ket x}{\sqrt{2^{|I^c|}}}\right)$$ 
where  $\ket{\phi^{\lambda_u, \alpha_u}_{s_u}}$ is the eigenvalue of $\mathcal O^{\lambda_u,\alpha_u}$ associated with the eigenvalue $(-1)^{s_u}$.
 
 Given an initial state $\ket \phi\in \mathbb C^{\{0,1\}^I}$ and $s\in \{0,1\}^{\comp O}$, the outcome of the computation is the state $A_s \ket \Psi$ (up to a normalisation factor), with probability $\bra \phi A^\dagger_sA_s\ket \phi$. In other words the MBQC implements the cptp-map\footnote{A completely positive trace-preserving map describes the evolution of a quantum system which state is represented by a density matrix. See for instance \cite{NC00} for details.} $\rho \mapsto \sum_{s\in \{0,1\}^{\comp O}} A_s \rho A^\dagger_s$. 
 
 Among all the possible measurement-based quantum computations, those which are deterministic are of peculiar importance. In particular, deterministic MBQC are those which are used to simulate quantum circuits (cornerstone of the proof that MBQC is a universal model of quantum computation), or to implement a quantum algorithm.  An MBQC $(G,I,O,\lambda, \alpha, \mathtt x, \mathtt z)$ is {\bf deterministic} if the output of the computation does not depend on the classical outcomes obtained during the computation: for any input state $\ket \phi\in \mathbb C^{\{0,1\}^I} $ and  branches $s,s'\in  \{0,1\}^{\comp O}$, $A_s\ket \phi$ and $A_{s'}\ket \phi$ are proportional. 
 
 Notice that the semantics of a deterministic MBQC $(G,I,O,\lambda, \alpha, \mathtt x, \mathtt z)$ is entirely defined by a single branch, e.g. the branch $A_{0^{|\comp O|}}$.  Moreover, this particular branch $A_{0^{|\comp O|}}$ is correction-free by construction (indeed all corrections are controlled by a signal, which is $0$ in this particular branch). As a consequence,  intuitively, when the evolution is deterministic, the corrections are only used to make the overall evolution deterministic but have no effect on the actual semantics of the evolution. Thus the correction can be abstracted away leading to the notion of {\bf abstract MBQC} $(G,I,O,\lambda, \alpha)$. There is however a caveat when the branch $A_{0^{|\comp O|}}$ is $0$: for instance  $M_1^{X,\pi}N_1N_2$ and $Z_2^{s_1}M_1^{X,\pi}N_1N_2$ are both deterministic\footnote{In both cases the unique measurement consists of measuring a qubit in state $\ket +$ according to the observable $-X$ which produces the signal $s_1=1$ with probability $1$.} and share the same abstract open graph, however they do not have the same semantics: the outcome of the former pattern is $\frac{\ket 0+\ket 1}{\sqrt 2}$, whereas the outcome of the latter is $\frac{\ket 0-\ket 1}{\sqrt 2}$. 
 
 To avoid these pathological cases and guarantee that the corrections can be abstracted away, a stronger notion of determinism has been introduced in \cite{BKMP07}: 
an MBQC is {\bf strongly deterministic}  when all the branches are not only proportional but equal up to a global phase. The strongness assumption guarantees that for any input state $\ket \phi$, $A_{0^{|\comp O|}}\ket \phi$ is non zero, and thus guarantees that the overall evolution is entirely described by the correction-free branch, or in other words by the knowledge of the abstract MBQC $(G,I,O,\lambda, \alpha)$.

Whereas deterministic MBQC are not necessarily invertible (e.g. $M^{(X,0)}_1N_2$ which maps any state $\ket \phi$ to the state $\ket +$), strongly deterministic MBQC correspond to the invertible deterministic quantum evolutions: they implement isometries ($\exists U : \mathbb C^{\{0,1\}^{I}}\to \mathbb C^{\{0,1\}^O}$ s.t. $U^\dagger U=I$ and $\forall s\in \{0,1\}^{|\comp O|}$, $\exists \theta$ s.t. $A_s=  {2^{-|\comp O|}}{e^{i\theta}}U$).

 We consider a variant of strong determinism which is robust to variation of the angles of measurements (which is a continuous parameter, so a priori subject to small variations in an experimental setting for instance), and to partial computation i.e., roughly speaking if one aborts the computation, the partial outcome does not depend on the branch of the computation.
 
 \begin{definition}[Robust Determinism] $(G,I,O,\lambda, \alpha,\mathtt x, \mathtt z)$ is robustly deterministic if for any lowerset 
   $S\subseteq O^c$ and for any $\beta:S\to [0,2\pi)$, $(G,I,O\cup S^c,\lambda\restrict{S}, \beta,\mathtt x \restrict S, \mathtt z \restrict S)$ is  strongly deterministic, where $S$ is a lowerset for the partial order induced by $\mathtt x$ and $\mathtt z$: $\forall v\in S, \forall u\in O^c$, $v\in \mathtt x(u)\cup \mathtt z(u) \Rightarrow u\in S$.
\end{definition}

The notion of \emph{robust determinism} we introduce is actually a short cut for \emph{uniformly strong and stepwise determinism} which has been already extensively studied in the context of measurement-based quantum computing \cite{BKMP07,DKPP09,MP08-icalp}. 

A central question in measurement-based quantum computation is to decide whether an abstract MBQC can be implemented deterministically: given $(G,I,O,\lambda, \alpha)$, does there exist correction strategies $\mathtt x, \mathtt z$ such that $(G,I,O,\lambda, \alpha,$ $\mathtt x, \mathtt z)$ is (robustly) deterministic? This question is related to the power of postselection in quantum computing: allowing postselection one can select the correc\-tion-free branch and thus implement any abstract MBQC $(G,I,O,\lambda, \alpha)$. Post-selection is a priori a non physical evolution, but in the presence of a correction strategy, postselection can be simulated using measurements and corrections.

The robustness assumption allows one to abstract away the angles and focus on the so-called {\bf open graph} $(G,I,O,\lambda)$ i.e.~essentially the initial entanglement. For which initial entanglement -- or in other words for which resource state --  a deterministic evolution can be performed? This is a fundamental question about the structures and the computational power of entanglement. 

Several graphical conditions for determinism have been introduced: causal flow, Generalized flow (Gflow) and Pauli Flow \cite{DK06,BKMP07,DKPP09}. These are graphical conditions on open graphs which are sufficient to guarantee the  existence of a robust deterministic evolution. Gflow has been proved to be a necessary condition for Pauli-free MBQC (i.e.~for any open graph $(G,I,O,\lambda)$ s.t. $\forall u\in O^c$, $|\lambda_u|=2$). The necessity of Pauli flow was an open question\footnote{In \cite{BKMP07}, an example of deterministic MBQC with no Pauli flow is given. This is however not a counter example to the necessity of the Pauli flow as the example is not robustly deterministic. More precisely not all the branches of computation occur with the same probability: with the notation of Figure 8 in \cite{BKMP07} if measurements of qubits 4,6,8 produce the outcome 0, then the measurement of qubit 10 produces the outcome 0 with probability 1.}. In this paper we show that Pauli flow fails to be necessary in general, but is however necessary for real MBQC, i.e.~when $\forall u\in O^c$, $\lambda_u\subseteq \{X,Z\}$. In the next section, we review the graphical sufficient conditions for determinism.

\subsection{Graphical Conditions for Determinism}

Several flow conditions for determinism have been introduced to guarantee robust determinism. Causal flow has been the first sufficient condition for determinism \cite{DK06}. 
This condition has been extended to Generalized flow (Gflow) and Pauli flow \cite {BKMP07}. Our first contribution is to provide a simpler description of the Pauli flow, equivalent to the original one  (see appendix \ref{app:prop1}): 

\begin{property}
\label{prop:sim}
$(G,I,O,\lambda)$ has a Pauli flow iff there exist a strict partial order $<$ over $O^c$ and $p : O^c \to 2^{I^c}$ s.t. $\forall u \in O^c$, 
\begin{eqnarray*}
		(c_X)\quad X\in \lambda_u&\Rightarrow & u\in \odd{p(u)}\setminus \left(  \bigcup_{\substack{v \geq u \\ v \notin O\cup\{u\}}} \odd{p(v)} \right)\\
(c_Y)\quad Y\in \lambda_u&\Rightarrow & u\in \codd{p(u)}\setminus \left(  \bigcup_{\substack{v \geq u \\ v \notin O\cup\{u\}}} \codd{p(v)} \right)\\
(c_Z)\quad Z\in \lambda_u&\Rightarrow & u\in {p(u)}\setminus \left(  \bigcup_{\substack{v \geq u \\ v \notin O\cup\{u\}}} {p(v)} \right)
\end{eqnarray*}
\noindent where $v\ge u$ iff $\neg (v<u)$
\end{property}

\begin{remark} Notice that the existence of a Pauli  flow forces the input qubits to be measured in the $\{X,Y\}$-plane: If $(G,I,O,\lambda)$ has a Pauli flow then for any $u\in I \cap  O^c$,  $u\notin p(u)$ since $p(u)\subseteq I^c$. It implies, according  to condition $(c_Z)$, that $Z\notin \lambda_u$.
\end{remark}

Gflow and Causal flows are special instances of Pauli flow: A Pauli flow is a Gflow when all measurements are performed in a plane (i.e. $\forall u$, $|\lambda_u|=2$); a Causal flow \cite{DK06} 
is nothing but a Gflow $(p,<)$ such that $\forall u, |p(u)|=1$. GFlow has been proved to be a necessary and sufficient condition for robust determinism:

\begin{theorem}[\cite{BKMP07}]\label{thm:gflow}
Given an abstract MBQC $(G,I,O,\lambda, \alpha)$ such that $\forall u\in O^c, |\lambda_u|=2$, 
{$(G,I,O,\lambda)$ has a GFlow $(p,<)$ if and only if  there exists $\mathtt x, \mathtt z$ extensive with respect to $<$  s.t. 
$(G,I,O,\lambda, \alpha, \mathtt x, \mathtt z)$ is robustly deterministic.}
\end{theorem}

Pauli flow is  the most general known sufficient condition for determinism for robust determinism:

\begin{theorem}[\cite{BKMP07}] \label{thm:Paulisuf}If $(G,I,O,\lambda)$ has a Pauli flow $(p,<)$,  then for any $ \alpha : O^c\to [0,2\pi)$, $(G,I,O,\lambda,\alpha, \mathtt x,\mathtt z)$ is robustly deterministic where $\forall u\in O^c$, 
\begin{align*}
\mathtt x(u)&= \{v\in p(u) ~|~ u< v\}\\
\mathtt z(u)&= \{v\in \odd{p(u)}~|~ u< v\}
\end{align*}
  
\end{theorem}

Is there a converse? 
This is the purpose of next section.

\section{Characterising Robust Determinism}
\label{sec:robdet}

In this section, we show the main  result of the paper: 
Pauli flow is necessary for robust determinism in the real case, i.e.~when all the measurements are in the $\{X,Z\}$-plane ($\forall u, \lambda_u\subseteq \{X,Z\}$). 

We investigate in the subsequent sections the consequences  of this result for real MBQC which is a universal model of quantum computation with several crucial applications. 

A \textbf{real open graph} $(G,I,O,\lambda)$ is an open graph such that $\forall u\in O^c, \lambda_u\subseteq \{X,Z\}$. We define similarly \textbf{real abstract MBQC} and \textbf{real MBQC}. Pauli flow conditions on real open graphs can be simplified as follows:

\begin{property} A real open graph $(G,I,O,\lambda)$  has a Pauli flow iff there exist a strict partial order $<$ over $O^c$ and $p : O^c \to 2^{I^c}$ s.t. $\forall u \in O^c$, 
\begin{eqnarray*}
		(i)\quad X\in \lambda_u&\Rightarrow & u\in \odd{p(u)}\setminus \left(  \bigcup_{\substack{v\geq u \\ v \notin O\cup\{u\} }} \odd{p(v)} \right)\\
(ii)\quad Z\in \lambda_u&\Rightarrow & u\in {p(u)}\setminus \left(  \bigcup_{\substack{v\geq u \\ v \notin O\cup\{u\} }} {p(v)} \right)\\
\end{eqnarray*}
\end{property}

\begin{theorem}
\label{thm:paulinec}
Given a real abstract MBQC $(G,I,O,\lambda, \alpha)$, 
{$(G,I,O,\lambda)$ has a Pauli flow $(p,\prec)$ if and only if  there exist $\mathtt x, \mathtt z$ extensive with respect to $\prec$  s.t. 
$(G,I,O,\lambda,$ $ \alpha, \mathtt x, \mathtt z)$ is robustly deterministic.}
\end{theorem}

The proof of Theorem \ref{thm:paulinec} is given in appendix. 
The proof 
is fundamentally different from the proof that Gflow is necessary for Pauli-free robust determinism (Theorem \ref{thm:gflow} in \cite{BKMP07}). Roughly speaking, the proof that Pauli flow is necessary goes as follows: first we fix the inputs to be either $\ket 0$ or $\ket +$ and all the measurements to be Pauli measurements (i.e.~if $\lambda_u=\{X,Z\}$ we fix the measurement of $u$ to be either $X$ or $Z$).  For each of these choices the computation can be described in the so-called stabilizer formalism which allows one to point out the constraints the corrections should satisfy for each of these particular choices of inputs and measurements. Then, as the corrections of a robust deterministic MBQC should not depend on the choice of the inputs and the angles of measurements, one can combine the constraints the corrections should satisfy and show that they coincide with the Pauli flow conditions.

\begin{remark}
We consider in this paper a notion of real MBQC which corresponds to a constraint on the measurements ($\forall u\in O^c, \lambda_u\in \{X,Z\}$), it can also be understood as an additional  constraint on the inputs: the input of the computation is in $\mathbb R^I$ instead of $\mathbb C^I$. This distinction might be important, for instance the pattern $M_1^YN_2$ is strongly deterministic on real inputs but not on arbitrary complex inputs. It turns out that the proof of Theorem \ref{thm:paulinec} only consider real inputs, and as a consequence is valid in both cases (i.e.~when both  inputs and measurements are real ; or when inputs are complex and measurements are in the $\{X,Z\}$-plane). 
\end{remark}

Pauli flow is necessary for real robust determinism. This property is specific to real measurements: Pauli flow is not necessary in general  even when the measurements are restricted to one of the other two planes of measurements.
  In the following $\{X,Y\}$-MBQC (resp. $\{Y,Z\}$-MBQC) refers to MBQC where all measurements are performed in the $\{X,Y\}$-plane (resp. $\{Y,Z\}$-plane).

\begin{property}
There exists   robustly deterministic  $\{X,Y\}$-MBQC (resp. $\{Y,Z\}$-MBQC) $(G,I,O,\lambda, \alpha, \mathtt x,\mathtt z)$  such that $(G,I,O,\lambda)$ has no Pauli flow $(p,\prec)$ where $\mathtt x$ and $\mathtt z$ are extensive with respect to $\prec$.

\end{property}

\begin{proof} We consider the pattern $\mathcal P = Z_3^{s_2}M_2^{X,0}X_2^{s_1}M_1^{\{X,Y\},\alpha}E_{1,2}E_{1,3}N_1N_2N_3$ which is an implementation of the $\{X,Y\}$-MBQC given in Fig \ref{fig:counterexample} (the other example is similar). Notice that the correction $X_2^{s_1}$ is useless as qubit $2$ is going to be measured according to $M^X$. Thus $\mathcal P$ has the same semantics as $\mathcal P' = Z_3^{s_2}M_2^{X,0}M_1^{\{X,Y\},\alpha}E_{1,2}E_{1,3}N_1N_2N_3$.  
 Notice in $\mathcal P'$ that the two measurements commute since there is no dependency between them, leading to the pattern $\mathcal P'' = M_1^{\{X,Y\},\alpha}Z_3^{s_2}M_2^{X,0}E_{1,2}E_{1,3}N_1N_2N_3$. It is easy to check that $\mathcal P''$ has a Pauli flow so is robustly deterministic. All but the stepwise property are transported by the transformations from $\mathcal P''$ to $\mathcal P$. Notice that $\mathcal P'$ is not stepwise deterministic as $M_1^{\{X,Y\},\alpha}E_{1,2}E_{1,3}N_1N_2N_3$ is not deterministic. However,  $\mathcal P$ enjoys the stepwise property  since $X_2^{s_1}M_1^{\{X,Y\},\alpha}E_{1,2}E_{1,3}N_1N_2N_3$ has a Pauli flow so is robustly deterministic. 
Finally, it is easy to show that the open graph has no Pauli flow $(p,\prec)$ such that $1\prec 2$, which is necessary to guarantee that $\mathtt x$ is extensive with respect to $\prec$.
\end{proof}

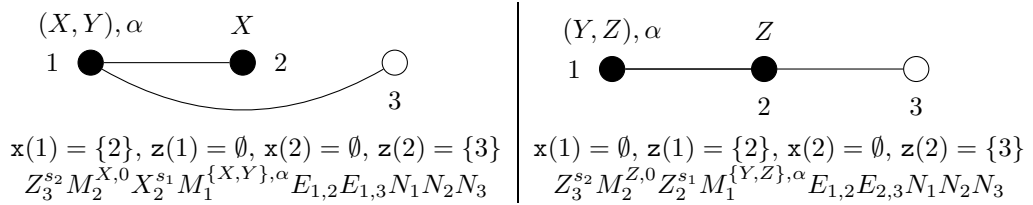
\begin{figure}[!h]
		\centering
    \setlength{\parindent}{0pt}
		\begin{tabular}{lc|c}
				&	
      \begin{tikzpicture}
			  \node[draw, circle, fill=black] (1) at (0,0) {};
			  \node[circle] () at (5,0) {};
			  \node[] (1h) at (0,0.5) {$(X,Y),\alpha$};
			  \node[] (1b) at (-0.5,0) {1};

			  \node[draw, circle] (3) at (4,0) {};
			  \node[] (2h) at (4,0.5) {};
			  \node[] (1b) at (4,-0.5) {3};
			  
			  \node[draw, circle, fill=black] (2) at (2,0) {};
			  \node[] (3h) at (2,0.5) {$X$};
			  \node[] (1b) at (2.5,0) {2};

				\draw (1) edge[out=-30,in=210,-] (3);
			  \draw[-,>=latex] (1) -- (2);
      \end{tikzpicture}

				&
      \begin{tikzpicture}
			  \node[draw, circle, fill=black] (1) at (0,0) {};
			  \node[circle] () at (5,0) {};
			  \node[] (1h) at (0,0.5) {$(Y,Z),\alpha$};
			  \node[] (1b) at (-0.5,0) {1};

			  \node[draw, circle] (3) at (4,0) {};
			  \node[] (2h) at (4,0.5) {};
			  \node[] (1b) at (4,-0.5) {3};
			  
			  \node[draw, circle, fill=black] (2) at (2,0) {};
			  \node[] (3h) at (2,0.5) {$Z$};
			  \node[] (1b) at (2,-0.5) {2};

			  \draw[-,>=latex] (1) -- (3);
			  \draw[-,>=latex] (1) -- (2);
      \end{tikzpicture}\\
				&
				$\mathtt x(1)=\{2\}$, $\mathtt z(1)=\emptyset$, $\mathtt x(2)=\emptyset$, $\mathtt z(2)=\{3\}$
				&
				$\mathtt x(1)=\emptyset$, $\mathtt z(1)=\{2\}$, $\mathtt x(2)=\emptyset$, $\mathtt z(2)=\{3\}$\\
				&
				$Z_3^{s_2}M_2^{X,0}X_2^{s_1}M_1^{\{X,Y\},\alpha}E_{1,2}E_{1,3}N_1N_2N_3$ 
				&
				$Z_3^{s_2}M_2^{Z,0}Z_2^{s_1}M_1^{\{Y,Z\},\alpha}E_{1,2}E_{2,3}N_1N_2N_3$

		\end{tabular}\\

		\caption{Robustly deterministic $\{X,Y\}$-MBQC and $\{Y,Z\}$-MBQC with no compatible Pauli flow. The two MBQC are described by means of there abstract MBQC $(G,I,O,\alpha)$ and the corrective maps $\mathtt x$ and $\mathtt z$. In both cases there is no input and the output is located on qubit $3$. A description using the measurement-pattern formalism is also provided (commands should be read from right to left). Notice that the only order that makes $\mathtt x$ and $\mathtt z$ extensive has to verify $1\prec 2$, and there is no Pauli flow for this order.}
		\label{fig:counterexample}
\end{figure}

\begin{remark}This is the last step of the proof of Theorem \ref{thm:paulinec} which fails with the examples of Figure \ref{fig:counterexample}. For instance in the $\{X,Y\}$-MBQC example, in both cases of Pauli measurements of qubit $1$ (according to $X$ or according to $Y$), 	a Pauli flow exists, sharing the same partial order $1\prec 2$. However the two Pauli flows are distinct and none of them is a Pauli flow when qubit $1$ is measured in the $\{X,Y\}$-plane.

\end{remark}

\begin{remark} The examples given in figure \ref{fig:counterexample} do have a Pauli flow but with a partial order not compatible with the order of measurements. It is important that the orders of the flow and the measurements coincide for guaranteeing that the depth of the flow (longest increasing sequence) corresponds to the depth of the MBQC. Because of the logarithmic separation between the quantum circuit model and MBQC in terms of depth (e.g. PARITY can be computed with a constant quantum depth MBQC but requires a logarithmic depth quantum circuit) \cite{BKP10}, it is also important that Pauli flow characterises not only the ability to perform a robust deterministic evolution, but characterizes also the depth of such evolution. There exists an efficient polynomial time which, given an open graph, compute a Gflow of optimal depth  (when it exists) \cite{MP08-icalp}, the existence of such an algorithm in the Pauli case is an open question.  
\end{remark}

\section{Applications: Computational Power of Real Bipartite MBQC}

In this section we focus on the real MBQC which  underlying graph are bipartite (real bipartite MBQC for short). Bipartite graphs  (or equivalently 2-colorable graphs)  play an important role in MBQC, the square grid is universal for quantum computing: any quantum circuit can be simulated by an MBQC whose underlying graph is a square grid. The brickwork graph \cite{BFK08} is bipartite and universal for $\{X,Y\}$-MBQC. Regarding real MBQC,  the (non bipartite) triangular grid is universal for real MBQC \cite{MP13} but there is no known universal  family of bipartite graphs. We show in this section that there is no universal family of bipartite graphs for real MBQC, by showing that any real bipartite MBQC can be done in constant depth.

\subsection{Real bipartite MBQC in constant depth}
In this section we show that real bipartite MBQC can always be parallelized:

\begin{theorem}\label{thm:constdepth}
All measurements of a robustly deterministic real bipartite MBQC can be performed in parallel. 
\end{theorem}

The rest of the section is dedicated to the proof of Theorem \ref{thm:constdepth}. According to Theorem \ref{thm:paulinec}, a real MBQC is robustly deterministic if and only if the underlying open graph has a Pauli flow. To prove that all the measurements can be performed in parallel in the bipartite case we point out the existence of a particular correction strategy which ensures that  each measurement is corrected using output qubits only. 
\begin{lemma}\label{lem:8} Given a bipartite graph $G$, $I, O\subseteq V(G) $ and $\lambda:O^c\to \{\{X\}, \{Z\},$ $ \{X,Z\}\}$, 
if $(G,I,O,\lambda)$ has a Pauli flow then there exists  $p:O^c\to 2^{I^c}$ s.t.:
\begin{eqnarray*}
\odd{p(u)} \setminus  (O\cup \lambda^{-1}(\{Z\})) & =& \{u\}\setminus  \lambda^{-1}(\{Z\})\\
{p(u)} \setminus  (O\cup \lambda^{-1}(\{X\})) & = &\{u\}\setminus  \lambda^{-1}(\{X\})
\end{eqnarray*}
\end{lemma}

\begin{proof}
		See appendix \ref{app:lemma8}
\end{proof}

This particular correction strategy corresponds to a king of \emph{super-normal form}. Indeed it is known that Gflow can be put into the so called $Z$- or $X$-normal form but not both at the same time (see \cite{Hamrit2015} for details). Lemma \ref{lem:8} shows, roughly speaking, that the Pauli flow in the real bipartite case can be put in both normal forms at the same time. 

\begin{proof}[Proof of Theorem \ref{thm:constdepth}]
			Given a robustly deterministic real bipartite MBQC $(G,I,O,\lambda, \alpha, \mathtt x, \mathtt z)$, according to Theorem \ref{thm:paulinec}, $(G,I,O,\lambda)$ has a Pauli flow, so according to Lemma \ref{lem:8} there exists $p$ s.t. $
\odd{p(u)} \setminus  (O\cup \lambda^{-1}(\{Z\}))  = \{u\}\setminus  \lambda^{-1}(\{Z\})$ and ${p(u)} \setminus  (O\cup \lambda^{-1}(\{X\}))  = \{u\}\setminus  \lambda^{-1}(\{X\})$. 
Notice that $(p,\emptyset)$ is a Pauli flow for $(G,I,O,\lambda)$, thus according to Theorem \ref{thm:Paulisuf}, $(G,I,O,\lambda, \alpha, \mathtt x', \mathtt z')$ is robustly deterministic where $\mathtt x'= u\mapsto  p(u)\setminus (\lambda^{-1}(\{X\}) \cup \{u\})$ and $\mathtt z'= u\mapsto \odd{p(u)}\setminus (\lambda^{-1}(\{Z\}) \cup \{u\})$. Both $(G,I,O,\lambda, \alpha, \mathtt x, \mathtt z)$ and $(G,I,O,\lambda, \alpha, \mathtt x', \mathtt z')$ implement the same computation, and $\forall u\in O^c$ $\mathtt x'(u)\subseteq O$ and $\mathtt z'(u)\subseteq O$ which implies that all measurements of the latter MBQC can be performed in parallel. 
\end{proof}

\subsection{Interactive proofs}
\label{subs:intpro}

The starting point of our work has been a sentence of McKague in \cite{McKague}. In the future work section, McKague wonders how his work could be used to build an interactive prover with only two provers. 
The problem that McKaque wants to solve is the following.
We imagine a classical verifier, which is a computer with classical resources, who wants to perform a computation using some non-communicating quantum provers. The quantum provers are computers with quantum resources. In fact, the classical verifier wants to achieve his computation using the quantum power of quantum provers.
In this model, the hard point to breakthrough is that we want the verifier to detect cheating behavior of some provers.
The model should guarantee the verifier that the result of the computation made by the provers is correct: if a prover has cheated and not computed what he was asked, the verifier should be able to detect it.
We specify that the provers, in this model, cannot communicate one with the others: each prover can try to cheat on his own but he does not have the power to do it by exchanging information with the others.
McKague, in \cite{McKague}, proves that it is possible to imagine a protocol in which the computation can be performed by the classical verifier using a polynomial number of quantum provers.
To achieve this goal, McKague uses two main tools, one of them being Measurement Based Quantum Computation in the $(X,Z)$ plane. Mhalla and Perdrix, in \cite{MP13}, prove that there exists a grid that enables to perform a universal computing in the $(X,Z)$ plane. Usually, the $(X,Y)$ plane, first known to allow universal computation is preferred.
In his work, McKague needs the $(X,Z)$ plane: to be able to detect cheating behavior, McKague needs to compute in the reals. The conjugation operation that can be performed in other planes is a problem to detect some cheatings.

In his future work section, McKague argues that  most his work could be used to improve his result to the use of only two provers. The main difficulty he points out is to build a bipartite graph to compute with. His self-testing skill, which is the second important tool of his work, can be applied only if the graph does not have any odd cycle.
Therefore, the question we wanted to answer was whether one could build a universal bipartite grid for the $(X,Z)$-plane. 
Our Theorem \ref{thm:constdepth} shows that in the real case a bipartite graph is not very powerful to compute: it is far from being universal. Therefore, at best, new skills will be needed to adapt McKague's method to interactive proofs with two provers.

\section{Conclusion and future work}

In this paper, we made substantial steps in understanding MBQC world. The first important one is this equivalence between being robustly deterministic and having a Pauli flow for a real-MBQC. Since it does not hold for $\{X,Y\}$- and $\{Y,Z\}$-planes, a natural question is how one  can modify the Pauli flow definition to obtain a characterisation of determinism in these cases? 
A bi-product of the characterisation of robust determinism for real MBQC is the low comutational power of real bipartite MBQC. It would be interesting to compare the computational power of real bipartite MBQC and of commuting quantum circuits. There are some good reasons to think that the power of real bipartite MBQC is exactly the same as those commuting quantum circuits.
Taking a global view of the MBQC domain, some advances we make in this paper, and a good direction for further research should be to better understand  the specificity of each plane in the power of the MBQC model and how the ability to perform a deterministic computation is linked to this power.
Finally, another open question is the existence of an efficient algorithm for deciding whether a given open graph has a Pauli flow, and which produces a Pauli flow of optimal depth when it exists. Such an algorithm exists for Gflow~\cite{MP08-icalp}. 

\bibliographystyle{plainurl}

\appendix

\section{Quantum Computing in a Nutshell}\label{QC}  
The state of a given a finite set (or register) $A$ of qubits is a unit vector $\ket \phi \in \mathbb C^{\{0,1\}^A}$. The so-called classical states $\{\ket x:x\in \{0,1\}^A\}$ of the register $A$ form an orthonormal basis $\mathbb C^{\{0,1\}^A}$, thus any state $\ket \phi$ of $A$ can be described as $\ket \phi = \sum_{x\in \{0,1\}^A}\alpha_x\ket x$ s.t. $\sum_{x\in \{0,1\}^A}|\alpha_x|^2 =1$.  
Given two distinct registers $A$ and $B$, if the state of $A$ is $\ket \phi = \sum_{x\in \{0,1\}^A}\alpha_x\ket x$ and  the state of $B$ is $\ket \psi = \sum_{y\in \{0,1\}^B}\beta_y\ket y$, then the state of the overall register $A\cup B$ is $\ket \phi \otimes \ket \psi = \sum_{x\in \{0,1\}^A,y\in \{0,1\}^B}\alpha_x\beta_y\ket {xy}$, where $xy$ is the concatenation of $x$ and $y$. 

The adjoint of a state $\ket \phi = \sum_{x\in \{0,1\}^A}\alpha_x\ket x\in \mathbb C^{\{0,1\}^A} $ is $\bra \phi  =(\ket{\phi})^\dagger =\sum_{x\in \{0,1\}^A}\alpha^
*_x\bra x\in \mathbb C^{\{0,1\}^A} \to 1$, where $\forall x,y\in \{0,1\}^A$, $\bra x \ket y = \delta_{x,y}$.  

Any quantum evolution can be decomposed into a sequence of \emph{isometries} and \emph{measurements}: An isometry $U : \mathbb C^{\{0,1\}^A} \to \mathbb C^{\{0,1\}^B}$ is a linear map s.t.~$U^\dagger U = I$, (i.e.~$\forall x\in \{0,1\}^A, (U\ket x)^\dagger (U\ket x) = 1$), which transforms the state $\ket \phi$ into $U\ket \phi$. Famous examples of isometries are the unitary evolutions which correspond to the case $|A|=|B|$. The simplest example of unitary transformations are the so-called one-qubit Pauli operators $X$, $Y$, $Z$: $X= \ket x\mapsto \ket {1-x}$, $Z = \ket x\mapsto (-1)^x \ket x$ and $Y = iXZ$.  An example of an isometry which is not a unitary evolution is, given a one-qubit state $\ket \psi \in \{0,1\}^{\{u\}}$ and a register $A$ s.t. $u\notin A$, the map $\ket \psi_u :  \mathbb C^{\{0,1\}^A} \to \mathbb C^{\{0,1\}^{A\cup \{u\}}} = \ket \phi \mapsto \ket \phi\otimes \ket{\psi}$ which consists of adding a qubit $u$ in the state $\ket \psi$ to the register $A$. 

A measurement is a fundamentally probabilistic evolution which produces a classical outcome and transforms the state of the quantum system. We consider in this paper only destructive measurements which means that the measured qubit is consumed by the measurement: measuring a qubit $u$ of a register $A$ transforms the state $\ket \phi \in \{0,1\}^A$ into a state $\ket \psi \in \{0,1\}^{A\setminus \{u\}}$. Moreover, we will consider only one-qubit measurements, also called local measurements. A 1-qubit measurement is characterised by an observable $\mathcal O$, i.e. an hermitian operator acting on one qubit. We assume $\mathcal O$ has two distinct eigenvalues $1$ and $-1$. Let $\ket {\phi_0}$ and $\ket {\phi_1}$ be the corresponding eigenvectors. 
A measurement according to $\mathcal O$ of a qubit $u$ of a register $A$ in the state $\ket \psi\in \mathbb C^{\{0,1\}^A}$ produces the classical outcome $0$ (resp. $1$) and the state  $\frac{\bra{\phi_0}_u \ket\psi}{\sqrt{\bra \phi \ket {\phi_0}_u\bra{\phi_0}_u \ket\psi}}$ (resp. $\frac{\bra{\phi_1}_u \ket\psi}{\sqrt{\bra \phi \ket {\phi_1}_u\bra{\phi_0}_u \ket\psi}}$) with probability $\bra \phi \ket {\phi_0}_u\bra{\phi_0}_u \ket\psi$  (resp. $\bra \phi \ket {\phi_1}_u\bra{\phi_1}_u \ket\psi$), where $\bra {\phi_1}_u : \mathbb C^{\{0,1\}^{A\cup \{u\}}} \to \mathbb C^{\{0,1\}^A}$ is the adjoint of $\ket {\phi_1}_u$. 

A quantum evolution composed of $k$ 1-qubit measurements and $n$ isometries (in any order) has $2^k$ possible evolutions and is hence represented by $2^k$ linear maps $L_s$ indexed by the possible sequences of classical outcomes. The quantum evolution should satisfy the condition $\sum_{s\in \{0,1\}^k}L^\dagger_sL_s = I$. It can be obtained as the composition of isometries and measurements as follows: a measurement is a pair $\{\bra {\phi_0}, \bra {\phi_1}\}$, an isometry $U$ is a singleton $\{U\}$ and the composition of two quantum evolutions is $\{L_s:s\in \{0,1\}^k\}\circ\{M_t:t\in \{0,1\}^m\} = \{L_sM_t :s\in \{0,1\}^k, t\in \{0,1\}^m\}$. 

A probability distribution of quantum states, say $\{(\ket {\phi_i}, p_i)\}_i$ can be represented as a density matrix $\rho = \sum_{i}p_i\ket {\phi_i}\bra{\phi_i}$. Two probability distributions of quantum states leading to the same density matrix are indistinguishable. A quantum evolution $\{L_s : s\in \{0,1\}^k\}$ transforms $\rho$ into $\sum_{s\in \{0,1\}^k}L_s\rho L^\dagger_s$.  

\section{Proof of property \ref{prop:sim}}
Pauli flow has been introduced in \cite{BKMP07}, as follows: 
\begin{definition}[Pauli Flow \cite{BKMP07}]
An open graph state $(G,I,O,\lambda)$ has \emph{Pauli flow} if there exists a map $p: O^c \rightarrow 2^{I^c}$ and a strict partial order $<$ over $O^c$ such that  $\forall u,v\in O^c$,
    \\---(P1)\ if $v\in p(u)$, $u\neq v$, and $\lambda_v\notin \{\{X\},\{Y\}\}$ then $u<v$,
    \\---(P2)\ if $v\le u$, $u\neq v$, and $\lambda_v\notin \{\{Y\},\{Z\}\}$  then $v\notin \odd{p(u)}$,
\\---(P3)\ if $v\le u$, $u\neq v$, and $\lambda_v=\{Y\}$ then $v\in p(u) \Leftrightarrow v\in \odd{p(u)}$,
\\---(P4)\ if $\lambda_u= \{X,Y\}$ then $u\notin p(u)$ and $u\in \odd{p(u)}$,
\\---(P5)\ if $\lambda_u=\{X,Z\}$ then $u\in p(u)$ and $u\in \odd{p(u)}$,
\\---(P6)\ if $\lambda_u=\{Y,Z\}$ then $u\in p(u)$ and $u\notin \odd{p(u)}$,
\\---(P7)\ if $\lambda_u=\{X\}$ then $u\in \odd{p(u)}$,
\\---(P8)\ if $\lambda_u=\{Z\}$ then $u\in p(u)$,
 \\---(P9)\ if $\lambda_u=\{Y\}$ then either:\;\; $u\notin p(u) \; \& \; u\in \odd{p(u)}$ \;\; or \;\; $u\in p(u) \; \& \; u\notin \odd{p(u)}$.
 
\noindent where $u\le v$ iff $\neg (v<u)$.
\end{definition}

\label{app:prop1}
		First of all, (P9) can be simplified to: if $\lambda_u=\{Y\}$ then $u\in \codd{p(u)}$.
		Let's now begin to rewrite the block (P4) to (P9).
		Using (P4), (P5) and (P7), we can say that $X \in \lambda_u \Rightarrow u \in \odd{p(u}$. Also, (P4), (P6) and (P9) enable us to show that $Y \in \lambda_u \Rightarrow u \in \codd{p(u)}$ and (P5), (P6) and (P8) that $Z \in \lambda_u \Rightarrow u \in p(u)$ is correct.
                Conversely, we can go back as easily to property (P4) to (P9) from $X\in \lambda_u \Rightarrow   u\in \odd{p(u)}$, $Y\in \lambda_u \Rightarrow   u\in \codd{p(u)}$ and $Z\in \lambda_u \Rightarrow   u\in {p(u)}$.

		To achieve the proof, we need to show that given a $u \in O^c$, for all $v \in O^c$, (P1), (P2) and (P3) are equivalent to the fact that if $v \leq u$ and $v \neq u$, then:
		\\---(Q1) $X \in \lambda_v \Rightarrow v \notin \odd{p(u)}$,
		\\---(Q2) $Y \in \lambda_v \Rightarrow v \notin \codd{p(u)}$,
		\\---(Q3) $Z \in \lambda_v \Rightarrow v \notin p(u)$.

		This equivalence is easier to prove once (P1), (P2) and (P3) are simplified to:
		\\---(P1') for $v \leq u$ and $v \neq u$, $\lambda_v \notin \{\{X\},\{Y\}\} \Rightarrow v \notin p(u)$,
		\\---(P2') for $v \leq u$ and $v \neq u$, $\lambda_v \notin \{\{Y\},\{Z\}\} \Rightarrow v \notin \odd{p(u)}$,
		\\---(P3') for $v \leq u$ and $v \neq u$, $\lambda_v = \{Y\}$, $v \in p(u) \Leftrightarrow v \in \odd{p(u)}$.

		The end of the proof is a proof by exhaustion.
		To prove (Q1) from (P1'), (P2') and (P3'), let's say that if $X \in \lambda_v$, then $\lambda_v$ is $\{X\}$, $\{X,Y\}$ or $\{X,Z\}$. (P2') enables us to conclude. To prove (Q2), let's say that if $Y \in \lambda_v$, then $\lambda_v$ is $\{Y\}$, $\{X,Y\}$ or $\{Y,Z\}$. In the first case, we can conclude from (P3'), in the other two, the combination of (P1') and (P2') do the trick. The third case goes the same way.

		Conversely, let's show that we can prove (P1') from (Q1), (Q2) and (Q3). The proof of (P2') and (P3') will follow the same sketch. If $\lambda_v \notin \{\{X\},\{Y\}\}$, then $\lambda_v$ is $\{Z\}$ or one of the three planes. If $\lambda_v$ is $\{Z\}$, $\{X,Z\}$ or $\{Y,Z\}$, then we get the result from (Q3). If $\lambda_v$ is $\{X,Y\}$, then we know from (Q1) that $v \notin \codd{p(u)}$ and from (Q2) that $v \notin \odd{p(u)}$: that sufficient to assure that $v \notin p(u)$. That ends the proof.
		
\section{Proof of Theorem \ref{thm:paulinec}}
\label{app:thempaulinec}
$[\Rightarrow]$: Theorem \ref{thm:Paulisuf}. $[\Leftarrow]$:  We order the vertices of $G$ according to the order of the measurements: $V=\{v_0, \ldots , v_{n-1}\}$ s.t. $v_i\prec v_j \Rightarrow i<j$. For any $k\in [0,n)$, let $V_k = \{v_k, \ldots, v_{n-1}\}$. 
For any $S\subseteq I$, let the input in $S$ be $\ket 0$ and those in $I\setminus S$ be $\ket +$. Moreover for any $u\in O^c$, let $M_u$ be a Pauli measurement of qubit $u$ with $M \in \lambda_u$. 
The initial state -- before the first measurement -- is $\ket {\phi_0} = \ket 0_S\otimes (\prod_{u,v \in S^c s.t. (u,v)\in G} \Lambda Z_{u,v})\ket +_{S^c}$. 
We are going to use some technical claims to build the proof, for which the proofs are given in appendix \ref{app:thempaulinec}. The following claims exhibit Pauli operators which depend on the measurements performed during the computation, and which stabilize the intermediate states obtained during the computation: 

~\\
\noindent{\bf Claim 1.} There exists $n$ independent\footnote{$P^{(0)}, \ldots, P^{(n-1)}$ are independent if none of these Pauli operators can be obtained as the product of the other ones, even up to a global phase.} Pauli operators $P^{(0)}, \ldots, P^{(n-1)}:\mathbb C^{\{0,1\}^V}$ $\to \mathbb C^{\{0,1\}^V}$ s.t. $\forall i\in [0,n)$,  $P^{(i)}\ket{\phi_0} = \ket {\phi_0}$ and $\forall j<i$, $M_{v_j}$ and $P^{(i)}$ commute. 
 
  ~\\
  
   \noindent{\bf [Proof of Claim 1]} For all $u\in V$, let $R^{(u)} = \begin{cases}Z_u&\text{if $u\in S$}\\ X_uZ_{N_G(u)}&\text{otherwise}\end{cases}$. 
The initial state $\ket {\phi_0}$ is stabilized by $\mathcal S = \langle R_u\rangle_{u \in V}$, i.e.~$\ket {\phi_0}$ is the unique state (up to an irrelevant global phase) such that $\forall u\in V(G)$, $R_u\ket {\phi_0} = \ket {\phi_0}$. 

We use the following Gauss-elimination-like algorithm to produce some new generators $(P^{(u)})_{u\in V}$ of $\mathcal S$ which satisfy that $\forall v_i\in O^c, \forall v_j\in V$, if $i<j$ then $M_{v_i}$ and $P^{(v_j)}$ commute:

\indent For all $u\in V, P^{(u)}\leftarrow R^{(u)}$.\\
\indent For all $i\in [0,|V|-1]$:\\
\indent\indent let $A = \{j ~|~i\le j  \text{ and $M_{v_i}$ and $P^{(v_j)}$ anticommute}\}$.\\ 
\indent\indent If $A\neq \emptyset$, let $i_0\in A$ \\
\indent\indent \indent for all $j\in A\setminus \{i_0\}$, $P^{(v_j)} \leftarrow P^{(v_j)}.P^{(v_{i_0})}$\\
\indent\indent \indent $P_{v_i}\leftrightarrow P_{v_{i_0}}$.
\hfill$\Box$
  
  ~\\
  
  \noindent{\bf Claim 2.} After $k$ measurements and the corresponding corrections, the state $\ket{\phi_k}$ of the system\footnote{To simplify the proof we assume that the measurements are non destructive, which means that after, say,  a $Z$-measurement the measured qubit remains and is either in state $\ket 0$ of $\ket 1$ depending on the outcome of the measurement.   As a consequence, for any $k$, $\ket{\phi_k}$ is a $n$-qubit state.}    satisfies: $\forall i<k$, $M_{v_i}\ket{\phi_k}=\pm\ket{\phi_k}$ and $\forall i\ge k$, $P^{(i)}\ket{\phi_k}=\ket{\phi_k}$. 
~\\

\noindent{\bf [Proof of Claim 2.]}
    Since the first $k$ qubits of $\ket {\phi_k}$ have been measured according to $M_{v_0},...,$ $M_{v_{k-1}}$, for any $i<k$, $M_{v_i}\ket{\phi_k} = (-1)^{s_i}\ket {\phi_k}$ where $s_i\in\{0,1\}$ is the classical outcome of measurement of qubit $v_i$. To prove that $\forall i\ge k$, $P^{(i)}\ket{\phi_k}=\ket{\phi_k}$, notice that if a quantum state is the fixpoint of some operator $P$, the measurement of this state according to an observable which commute with $P$ produces, whatever the classical outcome is, a quantum state which is also a fixpoint of $P$. Thus, since $P^{(i)}$ stabilizes the initial state $\ket{\phi_0}$ and commutes with the first $k$ measurements, it stabilizes $\ket{\phi_k}$.   \hfill$\Box$

~\\
  \noindent{\bf Claim 3.} For any $k$, and any $n$-qubit Pauli operator $P$ s.t.~$P\ket {\phi_k} =  \pm\ket {\phi_k}$, $\exists B_S\subseteq S^c$, $\exists D_S\subseteq S$, $\exists F_S\subseteq V_k^c$ s.t.~$P=\pm X_{B_S}Z_{\odd{B_S}\Delta D_S}\prod_{u\in F_S}M_u$.
~\\
  
\noindent{\bf [Proof of Claim 3.]}
   Claim 2 provides $n$ independent Pauli operators which stabilize $\ket {\phi_k}$ thus $P$ must be a product of these operators: $\exists F_S \subseteq V_k^c$, $\exists Q\subseteq [k,n)$, s.t. $P = \pm \prod_{u\in F_S} M_u \prod_{i\in Q_S} P^{(i)}$. Since each $P^{(i)}$ is, according to Claim 1, a product $\prod_{u\in \Gamma_i} R_u$ where of $R_u = \begin{cases}Z_u&\text{if $u\in S$}\\ X_uZ_{N_G(u)}&\text{otherwise}\end{cases}$. As a consequence, $P=\pm X_{B_S}Z_{\odd{B_S}\Delta D_S}\prod_{u\in F_S}M_u$, where $D_S = (\Gamma_k\Delta \ldots \Delta \Gamma_{n-1})\cap S$ and $B_S = (\Gamma_k\Delta \ldots \Delta \Gamma_{n-1})\cap S^c$. 
\hfill $\Box$

At some step $k$ of  the computation, by the strongness hypothesis, the two possible outcomes of the measurement according to $M_{v_{k}}$ occur with probability $1/2$. Thus, thanks to the stepwise determinism hypothesis, there exists a real state  $\ket \psi$ on qubits $V\setminus \{v_{k}\}$ and $\theta\in[0,2\pi)$ s.t.
\vspace{-0.2cm}$$\vspace{-0.2cm}\ket {\phi_k} = \frac1{\sqrt 2}(\ket \uparrow_{v_{k}} \otimes \ket \phi_{V\setminus \{v_{k}\}} + e^{i\theta} \ket \downarrow_{v_{k}} \otimes X_{\mathtt x(v_{k})} Z_{\mathtt z(v_{k})}  \ket \phi_{V\setminus \{v_{k}\}})$$
where $\ket{\uparrow}\in \{\ket 0, \frac{\ket 0+\ket 1}{\sqrt 2}\}$ and $\ket \downarrow\in \{\ket 1, \frac{\ket 0-\ket 1}{\sqrt 2}\}$ are the eigenvectors of $M_{v_{k}}$. Since $\ket {\phi_k}$, $\ket{\uparrow}$, $\ket{\downarrow}$ and $\ket{\phi}$ are real states, $e^{i\theta}=(-1)^r$ for some $r\in \{0,1\}$.  Let $T$ be the Pauli operator s.t. $T\ket \uparrow = - \ket \downarrow$ and $T\ket \downarrow = (-1)^{|\mathtt x(v_{k+1})\cap \mathtt z(v_{k+1}) |}\ket \uparrow$. 

		Since $(-1)^r M_{v_{k}}T_{v_{k}}X_{\mathtt x(v_{k})} Z_{\mathtt z(v_{k})} \ket {\phi_k} = \ket {\phi_k}$, according to claim 3, 
                $\exists B_S\subseteq S^c, D_S\subseteq S, F_S\subseteq V_{k}^c$ s.t.~$M_{v_{k}}T_{v_{k}}X_{\mathtt x(v_{k})} Z_{\mathtt z(v_{k})}  = \pm X_{B_S} Z_{Odd(B_S)\Delta D_S} \prod_{u\in F_S} M_u$. Thus, 
\vspace{-0.2cm}$$\vspace{-0.2cm}T_{v_{k}}X_{\mathtt x(v_{k})} Z_{\mathtt z(v_{k})}  = \pm X_{B_S} Z_{Odd(B_S)\Delta D_S} \prod_{u\in F'_S} M_u$$
with $ B_S\subseteq S^c, D_S\subseteq S, F'_S\subseteq V_{k+1}^c$. 

The  equation above, involving $\mathtt x(v_{k})$ and $\mathtt z(v_{k})$, is the main ingredient to recover the Pauli flow conditions. However, this equation depends a priori on the choice of the measurements and the initial sates. The following claims, proved in appendix,  show how to get rid of this dependency.

~\\
 \noindent{\bf Claim 4.}   $B_S, D_S$, and $F'_S$ do not depend on $S$. Therefore, using the notation $F$, $B$ and $D$ respectively, we can notice that $D = D_{\emptyset} =\emptyset$, and $B=B_I\subseteq I^c$.
 ~\\

\noindent{\bf [Proof of Claim 4.]}
For any $S,S'\subseteq I$, \\$X_{B_S\Delta B_{S'}} Z_{Odd(B_S\Delta B_{S'})\Delta D_S \Delta D_{S'}} \prod_{u\in F'_S\Delta F'_{S'}} M_u = \pm I$,
so  $ \prod_{u\in F'_S\Delta F'_{S'}} M_u = $\\ $\pm X_{B_S\Delta B_{S'}} Z_{Odd(B_S\Delta B_{S'})\Delta D_S \Delta D_{S'}} $.

\noindent Since all $M_u{\in} \{X,Z\}$, the product in the RHS of the latter equation should only produce $X$ and $Z$ (but no $XZ$), thus 
 $(B_S\Delta B_{S'}) \cap ({Odd(B_S\Delta B_{S'})\Delta D_S \Delta D_{S'}})=\emptyset$ which is equivalent to $(B_S\Delta B_{S'}) \cap (D_S \Delta D_{S'}) = (B_S\Delta B_{S'}) \cap ({Odd(B_S\Delta B_{S'})})$. Thus $|(B_S\Delta B_{S'}) \cap (D_S \Delta D_{S'})|=0\bmod 2$, since for any set $A$, $|A\cap Odd(A)| = 0\bmod 2$. 

To prove that $F'_S = F'_{S'}$ we exhibit a particular input state such that the initial state is an eigenvector of $\prod_{u\in F'_S\Delta F'_{S'}} M_u$. It implies that when the measurements are performed, the last measurement of $F'_S\Delta F'_{S'}$ is going to be deterministic and thus contradicts the strongness assumption. As a consequence $F'_S\Delta F'_{S'}$ must be empty. 
The input state is constructed as follows: The qubits in $(B_S\Delta B_{S'})\cap I$ are initialised in  $\ket +$, the others in $\ket 0$. Since $|(B_S\Delta B_{S'})\cap (D_S\Delta D_{S'})|=0\bmod 2$ there exists a partition of $(B_S\Delta B_{S'})\cap (D_S\Delta D_{S'})$ into pairs of qubits $P=\{(u_i,v_i)\}_i$. For each pair in $P$, $\Lambda Z$ is applied on the corresponding qubits. The input state is then a fixpoint of $ X_{(B_S\Delta B_{S'})\cap I} Z_{D_S \Delta D_{S'}}$, 
thus after the entangling stage the overall state (including input and non input qubits) is an eigenstate of $X_{B_S\Delta B_{S'}} Z_{Odd(B_S\Delta B_{S'})\Delta D_S \Delta D_{S'}}$ which implies that the measurement according to $\prod_{u\in F'_S\Delta F'_{S'}} M_u$ is not strong. As a consequence, $F_S=F_{S'}$ which implies $B_S=B_{S'}$ and $D_S=D_{S'}$. \hfill $\Box$

~\\
 \noindent{\bf Claim 5.} $F$ and $B$ do not depend on the choice of Pauli measurements.
~\\

\noindent{\bf[Proof of Claim 5.]}
If  $F$ and $B$ depend on the choice of the Pauli measurements then there exists two choices which differ on a single measurement and differ on at least one of the sets $F$, $B$. 
Let $(M_u)_{u\in O^c}$ and $(M'_u)_{u\in O^c}$ these two choices and $u_0\in O^c$ s.t. $\forall u\neq u_0$, $M_u=M'_u$ and $M_{u_0}\neq M'_{u_0}$. 
Let $B$, $F$ and $B'$, $F'$ the sets associated with these two choices of measurements. We have $X_{B\Delta B'} Z_{Odd(B\Delta B')} = \prod_{u\in F\setminus F'}M_u\prod_{u\in F'\setminus F}M'_u\prod_{u\in F\cap F'}M_uM'_u$. Notice that  $\prod_{u\in F\cap F'}M_uM'_u=$ $\begin{cases}I&\text{if $u_0\notin F\cap F'$}\\X_{u_0}Z_{u_0}&\text{if $u_0\in F\cap F'$}\end{cases}$. Since  $|({B\Delta B'})\cap Odd(B\Delta B')|=0\bmod 2$, we know that $\prod_{u\in F\cap F'}M_uM'_u=I$. \\As a consequence, $X_{B\Delta B'} Z_{Odd(B\Delta B')} = \prod_{u\in F\Delta F'}M_u$. Using similar arguments that has been used above, one can provide a particular input such that the state is a fixpoint of $\prod_{u\in F\Delta F'}M_u$ which implies that, if $F\Delta F'\neq \emptyset$ the last measurement of $F\Delta F'$ is deterministic, contradicting the strongness hypothesis. Thus $F=F'$, as a consequence $B=B'$. 
    \hfill$\Box$

    We are now able to build a Pauli flow. Since $F$ does not depend on the choice of the measurements, the basis of measurement of the qubits in $F$ must not vary, i.e.~$F\subseteq \lambda^{-1}(\{X\})\cup \lambda^{-1}(\{Z\})$. As a consequence, defining $F_X = F\cap\lambda^{-1}(\{X\}))$ and $F_Z = F\cap\lambda^{-1}(\{Z\}))$, we have $T_{v_{k}}X_{\mathtt x(v_{k})} Z_{\mathtt z(v_{k})} = \pm X_{B\Delta F_X} Z_{Odd(B)\Delta F_Z }$. 
Defining $p(v_{k}):=B$ one can double check that for any partial order $\prec$ with respect to which $\mathtt x$ and $\mathtt z$ are extensive, $(p,\prec)$ is a Pauli flow. Indeed, if $X\in \lambda_{v_{k}}$, T anti-commutes with $X$, thus $u \in Odd(p(v_{k}))$. Similarly if $Z\in \lambda_u$, $u\in p(v_{k})$.

Let $u\le v_{k}$, $u\neq v_{k}$. Since  $\mathtt x$ and $\mathtt z$ are extensive, it implies that $u\notin  \mathtt x(v_{k})\cup\mathtt z(v_{k})$. If $u\in Odd(p(v_{k}))$, $u\in F_Z$ so $u\in \lambda^{-1}(\{Z\})$ which implies that $X\notin \lambda_u$. So $X\in \lambda_u \Rightarrow u\notin Odd(p(v_{k}))$. Similarly $Z\in \lambda_u \Rightarrow u\notin p(v_{k})$
\hfill$\Box$

\section{Proof of Lemma  \ref{lem:8}}
\label{app:lemma8}

		Since $(G,I,O,\lambda)$ has a Pauli flow, there exists an order $<$ and $g:O^c\to 2^{I^c}$ s.t. $X\in \lambda_u \Rightarrow  u\in \odd{g(u)}\setminus \left( \displaystyle{ \bigcup_{\substack{v \geq u \\ v \neq u}} \odd{g(v)}} \right)$ and $Z\in \lambda_u \Rightarrow  u\in {g(u)}\setminus \left( \displaystyle{ \bigcup_{\substack{v \geq u \\ v \neq u}} {g(v)}} \right)$. Since $G$ is bipartite, let $V_0,V_1$ be a bipartition of $V(G)$ s.t. $V_0$ and $V_1$ are independent sets, and let $R_i = V_i\setminus O$ and $p: O^c\to 2^{I^c}$ be defined as follows: $\forall i\in \{0,1\}$, and  $\forall u\in R_i$, 
$$p(u):= g(u)\oplus \left(\!\bigoplus_{\substack{v\in \odd{g(u)}\setminus \{u\}\\ \textup{s.t.} X\in \lambda_v}}\!\!\!\!\!\!\!\!\!\!p_Z(v)\right) \oplus \left(\!\bigoplus_{\substack{v\in g(u)\setminus \{u\} \\\textup{s.t.} Z\in \lambda_v}}\!\!\!\!p_X(v)\right) $$
where $\forall i\in \{0,1\}\forall v\in R_i$, $p_X(v)=p(v)\cap V_i$, and $p_Z(v)=p(v)\cap V_{1-i}$.

The inductive definition of $p$ is well founded as the definition of $p(u)$ only depends on $p(v)$ with $u<v$.

Let $u$ be maximal for $<$, $\odd{p(u)}=\odd{g(u)}$. For any $v<u$, $v\in (O\cup \lambda^{-1}(\{Z\}))^c$ implies $X\in \lambda_v$, so according to the Pauli flow condition, $v\notin \odd{p(u)}$. Moreover, $u\notin \lambda^{-1}(\{Z\}) \Leftrightarrow X\in \lambda_u \Rightarrow u\in \odd{p(u)}$, thus $\odd{p(u)}\setminus ((O\cup \lambda^{-1}(\{Z\}))  = \{u\}\setminus  \lambda^{-1}(\{Z\})$. 
Similarly ${p(u)} \setminus  (O\cup \lambda^{-1}(\{X\}))  = \{u\}\setminus  \lambda^{-1}(\{X\})$. 

By induction, for a given $u\in O^c$, assume the property is satisfied for all $v\in O^c$ s.t. $u<v$ which implies:
\begin{eqnarray*}
\odd{p_X(v)} \setminus  (O\cup \lambda^{-1}(\{Z\})) & =& \emptyset\\
\odd{p_Z(v)} \setminus  (O\cup \lambda^{-1}(\{Z\})) & =& \{v\}\setminus  \lambda^{-1}(\{Z\})\\
{p_X(v)} \setminus  (O\cup \lambda^{-1}(\{X\})) & = &\{v\}\setminus  \lambda^{-1}(\{X\})\\
{p_Z(v)} \setminus  (O\cup \lambda^{-1}(\{X\})) & = &\emptyset
\end{eqnarray*}
\begin{align*}
&p(u)\setminus  (O\cup \lambda^{-1}(\{Z\})) = g(u)\setminus  (O\cup \lambda^{-1}(\{Z\})) \oplus \\
&  \left(\bigoplus_{\substack{v\in \odd{g(u)}\setminus \{u\} \\\textup{s.t.} X\in \lambda_v}}\!\!\!\!\!\!\!\!\!\!\!p_Z(v)\setminus  (O\cup \lambda^{-1}(\{Z\})) \right) \oplus \left(\bigoplus_{\substack{v\in g(u)\setminus \{u\}\\ \textup{s.t.} Z\in \lambda_v}}\!\!\!\!\!\!\!p_X(v)\setminus  (O\cup \lambda^{-1}(\{Z\})) \right)\\
&=g(u)\setminus  (O\cup \lambda^{-1}(\{Z\})) \oplus  \left(\bigoplus_{\substack{v\in g(u)\setminus \{u\}\\ \textup{s.t.} Z\in \lambda_v}}\{v\}\setminus  (O\cup \lambda^{-1}(\{Z\})) \right)\\
&=g(u)\setminus  (O\cup \lambda^{-1}(\{Z\}))  \oplus (g(u)\setminus \{u\})\setminus  (O\cup \lambda^{-1}(\{Z\})\\
&=\{u\}\setminus  (O\cup \lambda^{-1}(\{Z\})) =\{u\}\setminus  \lambda^{-1}(\{Z\}) 
\end{align*}
Similarly,\begin{align*}
&\odd{p(u)}\setminus  (O\cup \lambda^{-1}(\{Z\})) = \odd{g(u)}\setminus  (O\cup \lambda^{-1}(\{Z\})){\oplus}  \\
&\!\!\left(\!\!\bigoplus_{\substack{v\in \odd{g(u)}\setminus \{u\} \\\textup{s.t.} X\in \lambda_v}}\hspace{-0.8cm}\odd{p_Z(v)}\setminus  (O\cup \lambda^{-1}(\{Z\})) \!\!\right) \!\!{\oplus}\!\! \left(\!\!\bigoplus_{\substack{v\in g(u)\setminus \{u\} \\\textup{s.t.} Z\in \lambda_v}}\hspace{-0.45cm}\odd{p_X(v)}\setminus  (O\cup \lambda^{-1}(\{Z\}))\!\! \right)\\
&=\odd{g(u)}\setminus  (O\cup \lambda^{-1}(\{Z\})) \oplus  \left(\bigoplus_{\substack{v\in \odd{g(u)}\setminus \{u\}\\ \textup{s.t.} X\in \lambda_v}}\{v\}\setminus  (O\cup \lambda^{-1}(\{Z\})) \right)\\
&=\odd{g(u)}\setminus  (O\cup \lambda^{-1}(\{Z\}))  \oplus (\odd{g(u)}\setminus \{u\})\setminus  (O\cup \lambda^{-1}(\{Z\})\\
&=\{u\}\setminus  (O\cup \lambda^{-1}(\{Z\})) =\{u\}\setminus  \lambda^{-1}(\{Z\}) 
\end{align*}

\end{document}